\documentclass[a4paper,12pt]{article}
\usepackage[totalwidth=500pt,totalheight=680pt]{geometry}
\usepackage{graphicx}
\usepackage{color}
\usepackage{amsmath}
\usepackage{amssymb}
\usepackage{mathrsfs}         
\usepackage[all]{xy}

\usepackage{amsthm}

\usepackage{url}

\newtheorem{theorem}{Theorem} 
\newtheorem*{theorem*}{Theorem} 
\newtheorem{corollary}{Corollary} 
\newtheorem{proposition}{Proposition}

\theoremstyle{remark}

\newcommand{\W}{\textsc{W}}
\newcommand{\WSAT}{\textsc{W[SAT]}}

\newcommand{\FPT}{\textsc{FPT}}

\title{Parameterized Resolution with bounded conjunction}

\author{
Stefan Dantchev and Barnaby Martin\thanks{Supported by EPSRC grant EP/G020604/1.} \\ 
Engineering and Computing Sciences, Durham University, U.K.
}

\begin{document}

\maketitle

\begin{abstract}
We provide separations between the parameterized versions of Res$(1)$ (Resolution) and Res$(2)$. Using a different set of parameterized contradictions, we also separate the parameterized versions of Res$^*(1)$ (tree-Resolution) and Res$^*(2)$.
 \end{abstract}

\section{Introduction}

In a series of papers \cite{FOCS2007,Galesietal,BGL-SAT,BGLR} a program of \emph{parameterized proof complexity} is initiated and various lower bounds and classifications are extracted. The program generally aims to gain evidence that $\W[2]$ is different from $\FPT$ (though in the journal version \cite{FOCS2007journal} of \cite{FOCS2007} the former becomes $\WSAT$, and in the note \cite{PPCw1} $\W[1]$ is entertained). 
Parameterized proof (in fact, refutation) systems aim at refuting \emph{parameterized contradictions} which are pairs $(\mathcal{F},k)$ in which $\mathcal{F}$ is a propositional CNF with no satisfying assingment of weight $\leq k$. Several parameterized (hereafter often abbreviated as ``p-'') proof systems are discussed in \cite{FOCS2007,Galesietal,BGLR}. The lower bounds in  \cite{FOCS2007,Galesietal} and \cite{BGLR} amount to proving that the systems p-tree-Resolution, p-Resolution and p-bounded-depth Frege are not \emph{fpt-bounded}. Indeed, this is witnessed by the \emph{Pigeonhole principle}, and so holds even when one considers parameterized contradictions $(\mathcal{F},k)$ where $\mathcal{F}$ is itself an actual contradiction. Such parameterized contradictions are termed ``\emph{strong}'' in \cite{BGLR}, in which the authors suggest these are the only parameterized contradictions that should be considered, as general lower bounds -- even in p-bounded-depth Frege -- are trivial (see \cite{BGLR}). We sympathise with this outlook, but remind that there are alternative parameterized proof systems built from embedding (see \cite{FOCS2007,FOCS2007journal}) for which no good lower bounds are known even for general parameterized contradictions.

Kraj\'{\i}\v{c}ek introduced the system Res$(j)$ of Resolution-with-bounded-conjunction in \cite{WPHP-Krajicek}. The tree-like variant of this system is normally denoted Res$^*(j)$. Res$(j+1)$ incorporates Res$(j)$ and is ostensibly more powerful. This was demonstrated first for Res$(1)$ and Res$(2)$ in \cite{AtseriasBE02}, where a quasi-polynomial separation was given. This was improved in \cite{AutomatizabilityResolution}, until an exponential separation was given in \cite{Switching_small_restrictions}, together with like separations for Res$(j)$ and Res$(j+1)$, for $j >1$. Similar separations of Res$^*(j)$ and Res$^*(j+1)$ were given in \cite{Tree-like_Res(k)}. We are motivated mainly by the simplified and improved bounds of \cite{Rel-sep}, which use relativisations of the \emph{Least number principle}, LNP$_n$ and an ordered variant thereof, the \emph{Induction principle}, IP$_n$. The contradiction LNP$_n$ asserts that a partial $n$-order has no minimal element. In the literature it enjoys a myriad of alternative names: the  \emph{Graph Ordering Principle} GOP, \emph{Ordering Principle} OP and \emph{Minimal Element Principle} MEP. Where the order is total it is also known as TLNP and GT. The contradiction IP$_n$ uses the built-in order of $\{1,\ldots,n\}$ and asserts that: $1$ has property $P$, $n$ fails to have property $P$, and any number having property $P$ entails a larger number also having property $P$. Relativisation of these involves asserting that everything holds only on some non-empty subset of the domain (in the case of IP$_n$ we force $1$ and $n$ to be in this relativising subset).

In the world of parameterized proof complexity, we already have lower bounds for $\mathrm{p\mbox{-}Res}(j)$ (as we have for p-bounded-depth Frege), but we are still interested in separating levels $\mathrm{p\mbox{-}Res}(j)$. We are again able to use the \emph{relativised least number principle}, RLNP$_n$ to separate p-Res$(1)$ and p-Res$(2)$. Specifically, we prove that $(\mathrm{RLNP}_n,k)$ admits a polynomial-sized in $n$ refutation in Res$(2)$, but all p-Res$(1)$ refutations of $(\mathrm{RLNP}_n,k)$ are of size $\geq n^{\sqrt{k/4}}$. Although we use the same principle as \cite{Rel-sep}, the proof given there does not adapt to the parameterized world, and instead we look for inspiration to the proof given in \cite{BGLR} for the \emph{Pigeonhole principle}. For tree-Resolution, the situation is more complicated. The \emph{Relativised induction principle} RIP$_n$ of IP$_n$ admits fpt-bounded proofs in Res$^*(1)$, indeed of size $O(k!)$, therefore we are forced to alter this principle. Thus we come up with the \emph{Relativised vectorised induction principle} RVIP$_n$. We are able to show that $(\mathrm{RVIP}_n,k)$ admits $O(n^4)$ refutations in Res$^*(2)$, while every refutation in Res$^*(1)$ is of size $\geq n^{k/16}$. Note that both of our parameterized contradictions are ``strong'', in the sense of \cite{BGLR}. We go on to give extended versions of RVIP$_n$ and explain how they separate p-Res$^*(j)$ from p-Res$^*(j+1)$, for $j>1$.

This paper is organised as follows. After the preliminaries, we give our separations of p-Res$^*(j)$ from p-Res$^*(j+1)$ in Section~\ref{sec:tree-res} and our separation of p-Res$(1)$ from p-Res$(2)$ in Section~\ref{sec:res}. We then conclude with some remarks and open questions.

\section{Preliminaries}

A \emph{parameterized language} is a language $L\subseteq \Sigma^* \times \mathbb{N}$; in an instance $(x,k) \in L$, we refer to $k$ as the \emph{parameter}. A parameterized language is \emph{fixed-parameter tractable} (fpt - and in \FPT) if membership in $L$ can be decided in time $f(k).|x|^{O(1)}$ for some computable function $f$. If FPT is the parameterized analog of P, then (at least) an infinite chain of classes vye for the honour to be the analog of NP. The so-called W-hierarchy sit thus: $\FPT \subseteq \W[1] \subseteq \W[2] \subseteq \ldots \subseteq \WSAT$. For more on parameterized complexity and its theory of completeness, we refer the reader to the monographs \cite{DowneyFellows,FlumGrohe}. Recall that the \emph{weight} of an assignment to a propositional formula is the number of variables evaluated to true. Of particular importance to us is the parameterized problem \textsc{Bounded-CNF-Sat} whose input is $(\mathcal{F},k)$ where $\mathcal{F}$ is a formula in CNF and whose yes-instances are those for which there is a satisfying assignment of weight $\leq k$. \textsc{Bounded-CNF-Sat} is complete for the class $\W[2]$, and its complement (modulo instances that are well-formed formulae) \textsc{PCon} is complete for the class co-$\W[2]$. Thus, \textsc{PCon} is the language of \emph{parameterized contradictions}, $(\mathcal{F},k)$ \mbox{s.t.} $\mathcal{F}$ is a CNF which has no satisfying assignment of weight $\leq k$.

A \emph{proof system} for a parameterized language $L \subseteq \Sigma^* \times \mathbb{N}$ is a poly-time computable function $P:\Sigma^* \rightarrow\Sigma^*\times \mathbb{N}$ \mbox{s.t.} $\mathrm{range}(P)=L$. $P$ is \emph{fpt-bounded} if there exists a computable function $f$ so that each $(x,k)\in L$ has a proof of size at most $f(k).|x|^{O(1)}$. 
These definitions come from \cite{Galesietal,BGL-SAT,BGLR} and are slightly different from those in \cite{FOCS2007,FOCS2007journal} (they are less unwieldy and have essentially the same properties). The program of \emph{parameterized proof complexity} is an analog of that of Cook-Reckow \cite{Proof_Complexity_start}, in which one seeks to prove results of the form $\W[2]\neq$co-$\W[2]$ by proving that parameterized proof systems are not fpt-bounded. This comes from the observation that there is an fpt-bounded parameterized proof system for a co-$\W[2]$-complete $L$ iff $\W[2]=$co-$\W[2]$.

\emph{Resolution} is a refutation system for sets of clauses (formulae in CNF) $\Sigma$. It operates on clauses by the \emph{resolution} rule, in which from $(P \vee x)$ and $(Q \vee \neg x)$ one can derive $(P \vee Q)$ ($P$ and $Q$ are disjunctions of literals), with the goal being to derive the empty clause. The only other permitted rule in weakening -- from $P$ to derive $P \vee l$ for a literal $l$. We may consider a Resolution refutation to be a DAG whose sources are labelled by initial clauses, whose unique sink is labelled by the empty clause, and whose internal nodes are labelled by derived clauses. As we are not interested in polynomial factors, we will consider the \emph{size} of a Resolution refutation to be the size of this DAG. Further, we will measure this size of the DAG in terms of the number of variables in the clauses to be resolved -- we will never consider CNFs with number of clauses superpolynomial in the number of variables. 
We define the restriction of Resolution, \emph{tree-Resolution}, in which we insist the DAG be a tree.

The system of \emph{parameterized Resolution} \cite{FOCS2007} seeks to refute the parameterized contradictions of \textsc{PCon}. Given $(\mathcal{F},k)$, where $\mathcal{F}$ is a CNF in variables $x_1,\ldots,x_n$, it does this by providing a Resolution refutation of 
\begin{equation}
\mathcal{F}\cup \{\neg x_{i_1}\vee \ldots \vee \neg x_{i_{k+1}} : 1 \leq i_1 < \ldots < i_{k+1} \leq n \}.
\label{equ:W[2]}
\end{equation}
Thus, in parameterized Resolution we have built-in access to these additional clauses of the form $\neg x_{i_1}\vee \ldots \vee \neg x_{i_{k+1}}$, but we only count those that appear in the refutation.

A \emph{$j$-clause} is an arbitrary disjunction of conjunctions of size at most $j$. Res$(j)$ is a system to refute a set of $j$-clauses. There are four derivation rules. The $\wedge$-\emph{introduction rule} allows one to derive from $P \vee \bigwedge_{i\in I_{1}}l_{i}$ and $Q \vee \bigwedge_{i\in I_{2}}l_{i}$, $P \vee Q \vee \bigwedge_{i\in I_{1} \cup I_2} l_{i}$, provided $|I_1 \cup I_2|\leq j$ ($P$ and $Q$ are $j$-clauses). The \emph{cut} (or \emph{resolution}) rule allows one to derive from $P \vee \bigvee_{i\in I}l_{i}$ and $Q \vee \bigwedge_{i\in I}\neg l_{i}$, $P \vee Q$. Finally, the two weakening rules allow the derivation of  $P \vee \bigwedge_{i\in I}l_{i}$ from $P$, provided $|I| \leq j$, and $P \vee \bigwedge_{i\in I_{1}}l_{i}$ from $P \vee \bigwedge_{i\in I_{1} \cup I_2}l_{i}$. 

If we turn a Res$(j)$ refutation of a given set of $j$-clauses $\Sigma$ upside-down, i.e. reverse the edges of the underlying graph and negate the $j$-clauses on the vertices, we get a special kind of \emph{restricted branching $j$-program}. The restrictions are
as follows.
Each vertex is labelled by a $j$-CNF which partially represents the  
information
that can be obtained along any path from the source to the vertex (this is a \emph{record} in the parlance of \cite{proofs_as_games}).
Obviously, the (only) source is labelled with the constant $\top$.
There are two kinds of queries, which can be made by a vertex:
\begin{enumerate}
\item Querying a new $j$-disjunction, and branching on the answer: that is, from $\mathcal{C}$ and the question $\bigvee_{i\in I}l_{i}?$ we split on $\mathcal{C}\wedge\bigvee_{i\in I}l_{i}$ and  $\mathcal{C}\wedge\bigwedge_{i\in I} \neg l_{i}$.
\item Querying a known $j$-disjunction, and splitting it according to  
the answer: that is, from $\mathcal{C} \wedge \bigvee_{i\in I_1 \cup I_2}l_{i}$ and the question $\bigvee_{i\in I_1}l_{i}?$ we split on $\mathcal{C} \wedge \bigvee_{i\in I_1}l_{i}$ and  $\mathcal{C} \wedge \bigvee_{i\in I_2}l_{i}$.
\end{enumerate}
\noindent There are two ways of forgetting information. From $\mathcal{C}_1 \cup \mathcal{C}_2$ we can move to $\mathcal{C}_1$. And from $\mathcal{C} \wedge \bigvee_{i\in I_1}l_{i}$ we can move to $\mathcal{C} \wedge \bigvee_{i\in I_1 \cup I_2}l_{i}$. The point is that forgetting allows us to equate the information obtained along two different branches and thus to merge them into a single new vertex. A sink of the branching $j$-program must be labelled with the negation of a $j$-clause from $\Sigma$. Thus the branching $j$-program is supposed by default to solve the \emph{Search problem for $\Sigma$}: given an assignment of the variables, find a clause which is falsified under this assignment.

The equivalence between a Res$(j)$ refutation of $\Sigma$ and a branching $j$-program of the kind above is obvious. Naturally, if we allow querying single variables only, we get branching $1$-programs -- decision DAGs -- that correspond to Resolution. If we do not
allow the forgetting of information, we will not be able to merge distinct
branches, so what we get is a class of decision trees that correspond
precisely to the tree-like version of these refutation systems. These decision DAGs permit the view of Resolution as a game between a Prover and Adversary (originally due to Pudlak in \cite{proofs_as_games}). Playing from the unique source, Prover questions variables and Adversary answers either that the variable is true or false (different plays of Adversary produce the DAG). Internal nodes are labelled by conjunctions of facts (\emph{records} to Pudlak) and the sinks hold conjunctions that contradict an initial clause. Prover may also choose to forget information at any point -- this is the reason we have a DAG and not a tree. Of course, Prover is destined to win any play of the game -- but a good Adversary strategy can force that the size of the decision DAG is large, and many Resolution lower bounds have been expounded this way. 


We may consider any refutation system as a parameterized refutation system, by the addition of the clauses given in (\ref{equ:W[2]}). In particular, parameterized Res$(j)$ -- p-Res$(j)$ -- will play a part in the sequel.

\section{Separating p-Res$^*(j)$ and p-Res$^*(j+1)$}
\label{sec:tree-res}

The \emph{Induction Principle} $\mathrm{IP}_n$ (see \cite{Rel-sep}) is given by the following clauses:
\[
\begin{array}{cl}
P_1, \neg P_n \\
\bigvee_{j>i} S_{i,j} & i \in [n-1] \\
\neg S_{i,j} \vee \neg P_{i} \vee P_j & i \in [n-1], j \in [n]
\end{array}
\]
\noindent The \emph{Relativised Induction Principle} $\mathrm{RIP}_n$ (see \cite{Rel-sep}) is similar, and is given as follows.
\[
\begin{array}{cl}
R_1, P_1, R_n, \neg P_n \\
\bigvee_{j>i} S_{i,j} & i \in [n-1] \\
\neg S_{i,j} \vee \neg R_{i} \vee \neg P_i \vee R_j & i \in [n-1], j \in [n] \\
\neg S_{i,j} \vee \neg R_{i} \vee \neg P_i \vee P_j & i \in [n-1], j \in [n] \\
\end{array}
\]
\noindent The important properties of $\mathrm{IP}_n$ and $\mathrm{RIP}_n$, from the perspective of \cite{Rel-sep}, are as follows. $\mathrm{IP}_n$ admits refutation in $\mathrm{Res}^*(1)$ in polynomial size, as does $\mathrm{RIP}_n$ in $\mathrm{Res}^*(2)$. But all refutations of $\mathrm{RIP}_n$ in $\mathrm{Res}^*(1)$ are of exponential size. In the parameterized world things are not quite so well-behaved. Both $\mathrm{IP}_n$ and $\mathrm{RIP}_n$ admit refutations of size, say, $\leq 4 k!$ in $\mathrm{p\mbox{-}Res}^*(1)$; just evaluate variables $S_{i,j}$ from $i:=n-1$ downwards. Clearly this is an fpt-bounded refutation. We are forced to consider something more elaborate, and thus we introduce the \emph{Relativised Vectorised Induction Principle} $\mathrm{RVIP}_n$.
\[
\begin{array}{cl}
R_1, P_{1,1}, R_n, \neg P_{n,j} & j \in [n] \\
\bigvee_{l>i, m\in[n]} S_{i,j,l,m} & i,j \in [n] \\
\neg S_{i,j,l,m} \vee \neg R_{i} \vee \neg P_{i,j} \vee R_l & i \in [n-1], j,l,m \in [n] \\
\neg S_{i,j,l,m} \vee \neg R_{i} \vee \neg P_{i,j} \vee P_{l,m} & i \in [n-1], j,l,m \in [n] \\
\end{array}
\]

\subsection{Lower bound: A strategy for Adversary over $\mathrm{RVIP}_n$}

We will give a strategy for Adversary in the game representation of a $\mathrm{Res}^*(1)$ refutation. For convenience, we will assume that Prover never questions the same variable twice.

Information conceded by Adversary of the form $R_i, \neg R_i, P_{i,j}$ and $S_{i,j,l,m}$ makes the element $i$ \emph{busy} ($\neg P_{i,j}$ and $\neg S_{i,j,l,m}$ do not). 
The \emph{source} is the largest element $i$ for which there is a $j$ such that Adversary has conceded $R_i \wedge P_{i,j}$. Initially, the source is $1$. Adversary always answers $R_1, P_{1,1},$ $R_n, \neg P_{n,j}$ (for  $j \in [n]$), according to the axioms. 

If $i$ is below the source. When Adversary is asked $R_i$, $P_{i,j}$ or $S_{i,j,l,m}$, then he answers $\bot$.

If $i$ is above the source. When Adversary is asked $R_i$, or $P_{i,j}$, then he gives Prover a free choice unless: 1.) $R_i$ is asked when some $P_{i,j}$ was previously answered $\top$ (in this case $R_i$ should be answered $\bot$); or 2.) Some $P_{i,j}$ is asked when $R_{i}$ was previously answered $\top$ (in this case $P_{i,j}$ should be answered $\bot$). When Adversary is asked $S_{i,j,l,m}$, then again he offers Prover a free choice. If Prover chooses $\top$ then Adversary sets $P_{i,j}$ to $\bot$.

Suppose $i$ is the source. Then Adversary answers $P_{i,j}$ and $S_{i,j,k,l}$ as $\bot$, unless $R_i \wedge P_{i,j}$ witnesses the source. If $R_i \wedge P_{i,j}$ witnesses the source, then, if $k$ is not the next non-busy element above $i$, answer $S_{i,j,l,m}$ as $\bot$. If $k$ is the next non-busy element above $i$, then give $S_{i,j,l,m}$ a free choice, unless $\neg P_{l,m}$ is already conceded by Adversary, in which case answer $\bot$. 

Using this strategy, Adversary can not be caught lying until either he has conceded that $k$ variables are true, or he has given Prover at least $n-k$ free choices.

Let $T(p,q)$ be some monotone decreasing function that bounds the size of the game tree from the point at which Prover has answered $p$ free choices $\top$ and $q$ free choices $\bot$. We can see that $T(p,q) \geq T(p+a,q) + T(p,q+a) +1$ and $T(k,n-k)\geq 0$. The following solution to this recurrence can be found in \cite{FOCS2007journal}.
\begin{corollary}
Every $\mathrm{p\mbox{-}Res}^*(1)$ refutation of $\mathrm{RVIP}_n$ is of size $\geq n^{k/16}$.
\end{corollary}
We may increase the number of relativising predicates to define $\mathrm{RVIP}^r_n$ (note $\mathrm{RVIP}^1_n=\mathrm{RVIP}_n$).
\[
\begin{array}{cl}
R^1_1,\ldots,R^1_1, P_{1,1}, R^1_n,\ldots,R^r_n \neg P_{n,j} & j \in [n] \\
\bigvee_{l>i, m\in[n]} S_{i,j,l,m} & i,j \in [n] \\
\neg S_{i,j,l,m} \vee \neg R^1_{i} \vee \ldots \vee \neg R^r_{i} \vee \neg P_{i,j} \vee R^1_l & i \in [n-1], j,l,m \in [n] \\
\vdots \\
\neg S_{i,j,l,m} \vee \neg R^1_{i} \vee \ldots \vee \neg R^r_{i} \vee \neg P_{i,j} \vee R^r_l & i \in [n-1], j,l,m \in [n] \\
\neg S_{i,j,l,m} \vee \neg R^1_{i} \vee \ldots \vee \neg R^r_{i} \vee \neg P_{i,j} \vee P_{l,m} & i \in [n-1], j,l,m \in [n] \\
\end{array}
\]
\noindent We show how to adapt the previous argument in order to demonstrate the following.
\begin{corollary}
Every $\mathrm{p\mbox{-}Res}^*(j+1)$ refutation of $\mathrm{RVIP}{j}$ is of size $\geq n^{k/16}$.
\end{corollary}
\noindent We use essentially the same Adversary strategy in a branching $j$-program. We answer questions $l_1 \vee \ldots \vee l_j$ as either forced or free exactly according to the disjunction of how we would have answered the corresponding $l_i$s, $i \in [j]$, before (if one $l_i$ is free, then the disjunction is also free).
The key point is that once some positive disjunction involving some subset of $R^1_i,\ldots,R^r_i$ or $P_{i,j}$ (never all of these together, of course), is questioned then, on a positive answer to this, the remaining unquestioned variables of this form should be set to $\bot$.

\subsection{Upper bound: a $\mathrm{Res}^*(j+1)$ refutation of $\mathrm{RVIP}^j_n$}

Look at the simpler, but very similar, refutation of $\mathrm{RIP}_n$ in $\mathrm{Res}^*(2)$, of size $O(n^2)$, as appears in Figure~\ref{fig:1}.
\begin{figure}
\[
\xymatrix{
\neg R_n \vee \neg P_n ? \ar[d]_{\top} \ar[r]^{\bot} & \# & \\
\neg R_{n-1} \vee \neg P_{n-1} ? \ar[d]_{\top} \ar[r]^{\bot} & S_{n-1,n} ? \ar[d]_{\top} \ar[r]^{\bot} & \# \\
\vdots \ar[d]_\top & \# & \\
\neg R_1 \vee \neg P_1 ? \ar[d]_{\top} \ar[r]^{\bot} & S_{1,n} ? \ar[d]_{\top} \ar[r]^{\bot} & \cdots \ar[r]^{\bot} & S_{1,2} ? \ar[d]_{\top} \ar[r]^{\bot} & \# \\
\# & \# & & \# \\
}
\]
\label{fig:1}
\caption{Refutation of $\mathrm{RIP}_n$ in $\mathrm{Res}^*(2)$}
\end{figure}
\begin{proposition}
There is a refutation of $\mathrm{RVIP}^j_n$ in $\mathrm{Res}^*(j+1)$, of size $O(n^{j+3})$.
\end{proposition}
\begin{proof}
We give the branching program for $j:=1$ -- the generalisation is clear.
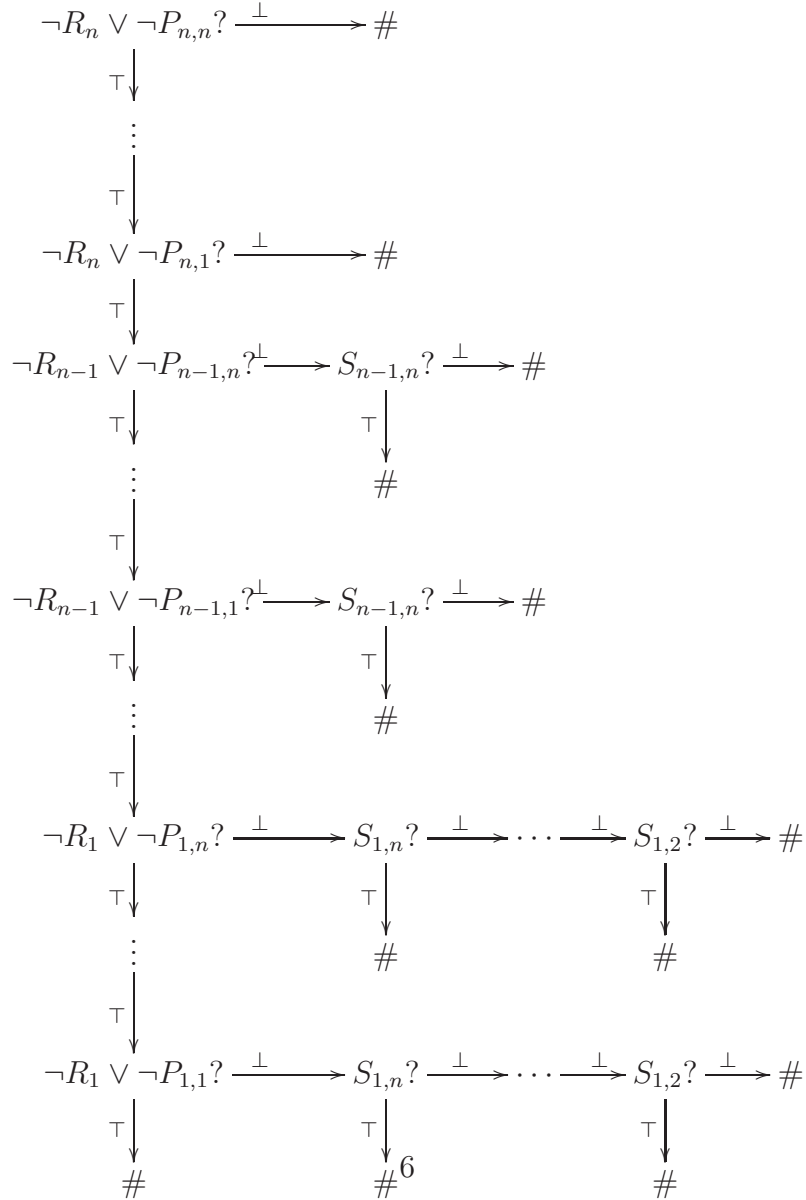
\[
\xymatrix{
\neg R_n \vee \neg P_{n,n} ? \ar[d]_{\top} \ar[r]^{\bot} & \# & \\
\vdots \ar[d]_{\top} \\
\neg R_n \vee \neg P_{n,1} ? \ar[d]_{\top} \ar[r]^{\bot} & \# & \\
\neg R_{n-1} \vee \neg P_{n-1,n} ? \ar[d]_{\top} \ar[r]^{\bot} & S_{n-1,n} ? \ar[d]_{\top} \ar[r]^{\bot} & \# \\
\vdots \ar[d]_\top & \# & \\
\neg R_{n-1} \vee \neg P_{n-1,1} ? \ar[d]_{\top} \ar[r]^{\bot} & S_{n-1,n} ? \ar[d]_{\top} \ar[r]^{\bot} & \# \\
\vdots \ar[d]_{\top} & \#\\
\neg R_1 \vee \neg P_{1,n} ? \ar[d]_{\top} \ar[r]^{\bot} & S_{1,n} ? \ar[d]_{\top} \ar[r]^{\bot} & \cdots \ar[r]^{\bot} & S_{1,2} ? \ar[d]_{\top} \ar[r]^{\bot} & \# \\
\vdots \ar[d]_\top & \# & & \# \\
\neg R_1 \vee \neg P_{1,1} ? \ar[d]_{\top} \ar[r]^{\bot} & S_{1,n} ? \ar[d]_{\top} \ar[r]^{\bot} & \cdots \ar[r]^{\bot} & S_{1,2} ? \ar[d]_{\top} \ar[r]^{\bot} & \# \\
\# & \# & & \# \\
}
\]
\end{proof}

\section{Separating p-Res$(1)$ and p-Res$(2)$}
\label{sec:res}

The \emph{Relativized Least Number Principle} $\mathrm{RLNP}_n$ is given by the following clauses:
\[
\begin{array}{cl}
\neg R_i \vee \neg L_{i,i} & i \in [n] \\
\neg R_i \vee \neg R_j \vee \neg R_k \vee \neg L_{i,j} \vee \neg L_{j,k} \vee L_{i,k} & i,j,k \in [n] \\  
\bigvee_{i \in [n]} S_{i,j} & j \in [n] \\
\neg S_{i,j} \vee \neg R_j \vee R_i & i,j \in [n] \\
\neg S_{i,j} \vee \neg R_j \vee \neg L_{i,j} & i, j \in [n] \\
R_n 
\end{array}
\]
The salient properties of $\mathrm{RLNP}_n$ are that it is polynomial to refute in $\mathrm{Res}(2)$, but exponential in $\mathrm{Res}(1)$ (see \cite{Rel-sep}). Polynomiality clearly transfers to fpt-boundedness in $\mathrm{p\mbox{-}Res}(2)$, so we address the lower bound for $\mathrm{p\mbox{-}Res}(1)$.

\subsection{Lower bound: A strategy for Adversary over $\mathrm{RLNP}_n$}

We will give a strategy for Adversary in the game representation of a $\mathrm{p\mbox{-}Res}(1)$ refutation. The argument used in \cite{Rel-sep} does not adapt to the parameterized case, so we instead use a technique developed for the Pigeonhole principle by Razborov in \cite{BGLR}.

Recall that a \emph{parameterized clause} is of the form $\neg v_1 \vee \ldots \vee \neg v_{k+1}$ (where each $v_i$ is some $R$ ,$L$ or $S$ variable). The $i,j$ appearing in $R_i$, $L_{i,j}$ and $S_{i,j}$ are termed \emph{co-ordinates}. We define the following \emph{random restrictions}. Set $R_n:=\top$. Randomly choose $i_0 \in [n-1]$ and set $R_{i_0}:=\top$ and $L_{i_0,n}=S_{i_0,n}:=\top$. Randomly choose $n-\sqrt{n}$ elements from $[n-1]\setminus {i_0}$, and call this set $\mathcal{C}$. Set $R_i := \bot$ for $i \in \mathcal{C}$. Pick a random bijection $\pi$ on $\mathcal{C}$ and set $L_{i,j}$ and $S_{i,j}$, for $i,j \in \mathcal{C}$, according to whether $\pi(j)=i$. Set $L_{i,j}=L_{j,i}:=\bot$, if $j \in \mathcal{C}$ and $i \in [n] \setminus (\mathcal{C} \cup \{i_0\})$.

What is the probability that a parameterized clause is \textbf{not} evaluated to true by the random assignment? We allow that each of $\neg R_n$, $\neg R_{i,0}$, $\neg L_{i_o,n}$ and $\neg S_{i_0,n}$ appear in the clause -- leaving $k+1-4=k-3$ literals, within must appear $\sqrt{(k-3)/4}$ distinct co-ordinates. The probability that some $\neg R_i$ is not true is $\leq \frac{\sqrt{n}}{n-\sqrt{n}}\leq \frac{2}{\sqrt{n}}$. The probability that some $\neg L_{i,j}$ is not true, where one of the co-ordinates $i,j$ is possibly mentioned before, is $\leq \frac{1}{\sqrt{n}} \frac{1}{n-\sqrt{n}} \cdot \frac{n-\sqrt{n}}{n} \leq \frac{2}{\sqrt{n}}$. Likewise with $\neg S_{i,j}$.  Thus we get that the probability that a parameterized clause is \textbf{not} evaluated to true by the random assignment is $\leq \frac{2}{\sqrt{n}}^{\sqrt{(k-3)/4}} \leq (n/4)^{-\sqrt{k-3}} \leq n^{-\sqrt{k/4}}$.

Now we are ready to complete the proof. Suppose fewer than $n^{\sqrt{k/4}}$ parameterized clauses appear in a $\mathrm{p\mbox{-}Res}(1)$ refutation of $\mathrm{RLNP}_n$, then there is a random restriction as per the previous paragraph that evaluates all of these clauses to true. What remains is a $\mathrm{Res}(1)$ refutation of $\mathrm{RLNP}_{\sqrt{n}}$, which must be of size larger than $n^{\sqrt{k/4}}$ itself, for $n$ sufficiently large (see \cite{Rel-sep}). Thus we have proved.
\begin{theorem}
Every $\mathrm{p\mbox{-}Res}(1)$ refutation of $\mathrm{RLNP}_n$ is of size $\geq n^{\sqrt{k/4}}$.
\end{theorem}

\section{Concluding remarks}

It is most natural when looking for separators of $\mathrm{p\mbox{-}Res}^*(1)$ and $\mathrm{p\mbox{-}Res}^*(2)$ to look for CNFs, like $\mathrm{RVIP}_n$ that we have given. $\mathrm{p\mbox{-}Res}^*(2)$ is naturally able to process $2$-clauses and we may consider $\mathrm{p\mbox{-}Res}^*(1)$ acting on $2$-clauses, when we think of it using any of the clauses obtained from those $2$-clauses by distributivity. In this manner, we offer the following principle as being fpt-bounded for $\mathrm{p\mbox{-}Res}^*(2)$ but not fpt-bounded for $\mathrm{p\mbox{-}Res}^*(1)$. Consider the two axioms $\forall x (\exists y \neg S(x,y) \wedge T(x,y)) \vee P(x)$ and $\forall x,y T(x,y) \rightarrow S(x,y)$. This generates the following system $\Sigma_{PST}$ of $2$-clauses.
\[
\begin{array}{cl}
P_i \vee \bigvee_{j \in [n]} (\neg S_{i,j} \wedge T_{i,j}) & i \in [n] \\
\neg T_{i,j} \vee S_{i,j} & i,j \in [n]
\end{array}
\]
\noindent Note that the expansion of $\Sigma_{PST}$ to CNF makes it exponentially larger. It is not hard to see that $\Sigma_{PST}$ has refutations in $\mathrm{p\mbox{-}Res}^*(2)$ of size $O(kn)$, while any refutation in $\mathrm{p\mbox{-}Res}^*(1)$ will be of size $\geq n^{k/2}$.

All of our upper bounds, \mbox{i.e.} for both $\mathrm{RVIP}_n$ and $\mathrm{RLNP}_n$, are in fact polynomial, and do not depend on $k$. That is, they are rather more than fpt-bounded. If we want examples that depend also on $k$ then we may enforce this easily enough, as follows. For a set of clauses $\Sigma$, build a set of clauses $\Sigma'_k$ with new propositional variables $A$ and $B_1,B'_1,\ldots,B_{k+1},B'_{k+1}$. From each clause $\mathcal{C} \in \Sigma$, generate the clause $A \vee \mathcal{C}$ in $\Sigma'_k$. Finally, augment $\Sigma'_k$ with the following clauses: $\neg A \vee B_1 \vee B'_1$, \ldots, $\neg A \vee B_{k+1} \vee B'_{k+1}$. If $\Sigma$ admits refutation of size $\Theta(n^c)$ in $\mathrm{p\mbox{-}Res}^*(j)$ then $(\Sigma'_k,k)$ admits refutation of size $\Theta(n^c+2^{k+1})$. The parameterized contradictions so obtained are no longer ``strong'', but we could even enforce this by augmenting instead a Pigeonhole principle from $k+1$ to $k$.

It is hard to prove p-Res$(1)$ lower bounds for parameterized
$k$-clique on a random graph \cite{BGL-SAT}, but we now introduce a contradiction that looks similar but for which lower bounds should be easier.
It is a variant of the Pigeonhole principle which could give us another very natural separation of $\mathrm{p\mbox{-}Res}(1)$ from $\mathrm{p\mbox{-}Res}(2)$. Define the contradiction PHP$_{k+1,n,k}$, on variables $p_{i,j}$ ($i \in [k+1]$ and $j \in [n]$) and $q_{i,j}$ ($i \in [n]$ and $j \in [k]$), and with clauses:
\[
\begin{array}{ll}
\neg p_{i,j} \vee \neg p_{l,j} & i \neq l \in [k+1]; j \in [n] \\
\neg q_{i,j} \vee \neg q_{l,j} & i \neq l \in [n]; j \in [k] \\
\bigvee_{\lambda \in [n]} p_{i,\lambda} & i \in [k] \\
\neg p_{i,j} \vee \bigvee_{\lambda \in [k]} q_{j,\lambda} & j \in [n] \\
\end{array}
\]
\noindent We conjecture that this principle, which admits fpt-bounded refutation in $\mathrm{p\mbox{-}Res}(2)$, does not in $\mathrm{p\mbox{-}Res}(1)$.

Finally, we leave open the technical question as to whether suitably defined, further-relativised versions of RLNP$_n$ can separate $\mathrm{p\mbox{-}Res}(j)$ from $\mathrm{p\mbox{-}Res}(j+1)$. We conjecture that they can.


\begin{thebibliography}{10}

\bibitem{AutomatizabilityResolution}
A.~Atserias and M.~Bonet.
\newblock On the automatizability of resolution and related propositional proof
  systems.
\newblock In {\em 16th Annual Conference of the European Association for
  Computer Science Logic}, 2002.

\bibitem{AtseriasBE02}
Albert Atserias, Maria~Luisa Bonet, and Juan~Luis Esteban.
\newblock Lower bounds for the weak pigeonhole principle and random formulas
  beyond resolution.
\newblock {\em Inf. Comput.}, 176(2):136--152, 2002.

\bibitem{Galesietal}
Olaf Beyersdorff, Nicola Galesi, and Massimo Lauria.
\newblock Hardness of parameterized resolution.
\newblock Technical report, ECCC, 2010.

\bibitem{BGL-SAT}
Olaf Beyersdorff, Nicola Galesi, and Massimo Lauria.
\newblock Parameterized complexity of dpll search procedures.
\newblock In {\em Theory and Applications of Satisfiability Testing - SAT 2011
  - 14th International Conference, SAT 2011}, pages 5--18, 2011.

\bibitem{BGLR}
Olaf Beyersdorff, Nicola Galesi, Massimo Lauria, and Alexander~A. Razborov.
\newblock Parameterized bounded-depth frege is not optimal.
\newblock In {\em Automata, Languages and Programming - 38th International
  Colloquium, ICALP (1) 2011.}, pages 630--641, 2011.

\bibitem{Proof_Complexity_start}
S.~Cook and R.~Reckhow.
\newblock The relative efficiency of propositional proof systems.
\newblock {\em Journal of Symbolic Logic}, 44(1):36--50, March 1979.

\bibitem{Rel-sep}
S.~Dantchev.
\newblock Relativisation provides natural separations for resolution-based
  proof systems.
\newblock In {\em Computer Science - Theory and Applications, First
  International Computer Science Symposium in Russia, CSR 2006, St. Petersburg,
  Russia, June 8-12, 2006, Proceedings}, volume 3967 of {\em Lecture Notes in
  Computer Science}, pages 147--158. Springer, 2006.

\bibitem{FOCS2007}
Stefan Dantchev, Barnaby Martin, and Stefan Szeider.
\newblock Parameterized proof complexity.
\newblock In {\em 48th IEEE Symp. on Foundations of Computer Science}, pages
  150--160, 2007.

\bibitem{FOCS2007journal}
Stefan Dantchev, Barnaby Martin, and Stefan Szeider.
\newblock Parameterized proof complexity.
\newblock {\em Computational Complexity}, 20, 2011.

\bibitem{DowneyFellows}
Rodney~G. Downey and Michael~R. Fellows.
\newblock {\em Parameterized Complexity}.
\newblock Monographs in Computer Science. Springer Verlag, 1999.

\bibitem{Tree-like_Res(k)}
J.L. Esteban, N.~Galesi, and J.~Mesner.
\newblock On the complexity of resolution with bounded conjunctions.
\newblock In {\em Proceedings of the 29th International Colloquium on Automata,
  Languages and Programming}, 2002.

\bibitem{FlumGrohe}
J\"{o}rg Flum and Martin Grohe.
\newblock {\em Parameterized Complexity Theory}, volume XIV of {\em Texts in
  Theoretical Computer Science. An EATCS Series}.
\newblock Springer Verlag, 2006.

\bibitem{WPHP-Krajicek}
J.~Kraj\'{\i}\^{c}ek.
\newblock On the weak pigeonhole principle.
\newblock {\em Fundamenta Mathematica}, 170:123--140, 2001.

\bibitem{PPCw1}
Barnaby Martin.
\newblock Parameterized proof complexity and {W[1]}.
\newblock CoRR: arxiv.org/abs/1203.5323, 2012.
\newblock Submitted to Information Processing Letters.

\bibitem{proofs_as_games}
P.~Pudl\'{a}k.
\newblock Proofs as games.
\newblock {\em American Mathematical Monthly}, pages 541--550, June-July 2000.

\bibitem{Switching_small_restrictions}
N.~Segerlind, S.~Buss, and R.~Impagliazzo.
\newblock A switching lemma for small restrictions and lower bounds for $k$-dnf
  resolution.
\newblock In {\em Proceedings of the 43rd annual symposium on Foundations Of
  Computer Science}. IEEE, November 2002.

\end{thebibliography}
\end{document}